\documentclass[prx,onecolumn,english,superscriptaddress,floatfix,longbibliography,nofootinbib]{revtex4}

\pdfoutput=1

\usepackage[T1]{fontenc}
\usepackage[utf8]{inputenc}
\usepackage{caption}
\usepackage{amsmath, amsthm, amssymb,amscd, mathrsfs, amsfonts, mathtools,tikz-cd}
\usepackage{quantikz}
\usepackage{float} 
\usepackage{dsfont}
\usepackage{relsize}
\usepackage{lipsum}
\usepackage{pbox}
\usepackage[all]{xy}

\usepackage{xcolor}
\usepackage[bookmarks=true,
bookmarksnumbered=true, breaklinks=true,
pdfstartview=FitH, hyperfigures=false,
plainpages=false, naturalnames=true,
colorlinks=true,pagebackref=true,
pdfpagelabels]{hyperref}

\usepackage{tikz}
\usepackage{pgfplots}

\pgfplotsset{compat=1.10}
\usepgfplotslibrary{fillbetween}
\usetikzlibrary{patterns}
\usetikzlibrary{matrix}
\usepackage{pbox}
\usepackage[all]{xy}

\def\commutatif{\ar@{}[rd]|{\circlearrowleft}}

\newtheorem{thm}{Theorem}[section]

\newtheorem{prop}[thm]{Proposition}
\newtheorem{lem}[thm]{Lemma}
\newtheorem{cor}[thm]{Corollary}

\theoremstyle{definition}
\newtheorem{defn}[thm]{Definition}

\theoremstyle{remark}
\newtheorem{rmk}[thm]{Remark}

\newtheorem{ex}[thm]{Example}

\allowdisplaybreaks

\usepackage{bbm}

\hypersetup{
	colorlinks = true,
	urlcolor = blue,
	linkcolor = blue,
	citecolor = red,
	pdfpagemode = UseNone
}

\usepackage{youngtab}
\usepackage{ytableau}
\begin{document}


\author{Boris Tsvelikhovskiy} 
\affiliation{Department of Mathematics, University of California, Riverside, CA, USA} 

\author{Ilya Safro} 
\affiliation{Department of Computer and Information Sciences, University of Delaware, Newark, DE, USA} 

\author{Yuri Alexeev} 
\affiliation{Computational Science Division, Argonne National Laboratory, Argonne, IL, USA}

\title{Equivariant QAOA and the Duel of the Mixers}

    \begin{abstract}\textbf{Abstract.} Constructing an optimal mixer for Quantum Approximate Optimization Algorithm (QAOA)  Hamiltonian is crucial for enhancing the performance of QAOA in solving combinatorial optimization problems. We present a systematic methodology for constructing the QAOA tailored mixer Hamiltonian, ensuring alignment with the inherent symmetries of classical optimization problem objectives. The key to our approach is to identify an operator that commutes with the action of the group of symmetries on the QAOA underlying Hilbert space and meets the essential technical criteria for effective mixer Hamiltonian functionality.

We offer a construction method specifically tailored to the symmetric group $S_d$, prevalent in a variety of combinatorial optimization problems. By rigorously validating the required properties, providing a concrete formula and corresponding quantum circuit for implementation, we establish the viability of the proposed mixer Hamiltonian. Furthermore, we demonstrate that the classical mixer $B$ commutes only with a subgroup of $S_d$ of significantly smaller order than the group itself, enhancing the efficiency of the proposed approach.

To evaluate the effectiveness of our methodology, we compare two QAOA variants utilizing different mixer Hamiltonians—conventional $B=\sum X_i$ and the newly proposed $H_M$ — in edge coloring and graph partitioning problems across various graphs. We observe statistically significant differences in mean values, with the new variant consistently demonstrating superior performance across multiple independent simulations. Additionally, we analyze the phenomenon of poor performance in alternative warm-start QAOA variants, providing a conceptual explanation supported by recent literature findings.
\vspace{0.1in}

\hspace{-0.15in}\textbf{Keywords:} quantum approximate optimization algorithm, mixer Hamiltonians, warm-start QAOA
\end{abstract}

\maketitle 

\pgfkeys{/pgfplots/scale/.style={
  x post scale=#1,
  y post scale=#1,
  z post scale=#1}
}
\pgfkeys{/pgfplots/axis labels at tip/.style={
    xlabel style={at={(current axis.right of origin)}, xshift=1.5ex, anchor=center},
    ylabel style={at={(current axis.above origin)}, yshift=1.5ex, anchor=center}}
}

\section{Introduction}

In this paper we consider the optimization problem of finding extremal values of a function $F: \mathbb{D}^n \rightarrow \mathbb{R}$, where $\mathbb{D}^n$ represents the set of $n$-element $d$-ary strings and $\mathcal{S}$ is the group of permutations acting on these $d^n$ elements.
The Quantum Approximate Optimization Algorithm (QAOA), proposed in \cite{QAOA}, is a widely used approach for solving the quantum version of the optimization problem. This approach is considered as one of the main candidates to demonstrate practical quantum advantage in future in several areas \cite{herman2023quantum}. Consequently, there is a growing interest to enhance its performance. To bridge the classical and quantum realms, one employs the following correspondences:

\begin{itemize}
    \item $\mathbb{D}^n \rightsquigarrow$ vector space $W$ of dimension $d^n$ with basis $\{v_x\}$ indexed by elements $x\in\mathbb{D}^n$,
    \item Objective function $F \rightsquigarrow$ linear operator $H_P$  acting on $W$,
    \item Minima of $F$ on $\mathbb{D}^n \rightsquigarrow$ lowest energy states of $H_P$ in $W$.
\end{itemize}

Here, the Hamiltonian $H_P$ represents the objective function $F$, meaning it satisfies the equation $H_P(v_{x})=F(x)v_x$ for any string $x\in\mathbb{D}^n$. Another important component of the QAOA approach is an operator referred to as the mixer Hamiltonian $H_M$. This operator plays a pivotal role in the optimization process, as it possesses an easily identifiable ground state, which aids in initializing the optimization process.

The QAOA algorithm involves a multistep transformation of $H_M$ into $H_P$, aiming to obtain a lowest energy state for the latter Hamiltonian. This is achieved by alternately applying exponentials of $H_M$ and $H_P$, with the number of iterations denoted by $p$ (known as QAOA depth). We express this transformation as:
\begin{equation}
    \mathfrak{Q}_p = e^{-i\beta_1 H_M}e^{-i\gamma_1 H_P}\ldots e^{-i\beta_p H_M}e^{-i\gamma_p H_P}.
\end{equation}
The algorithm concludes with a measurement of the resulting state in the standard basis.

While the problem Hamiltonian $H_P$ is uniquely determined by the classical original problem (unless it is decided to be changed, e.g., by sparsification \cite{liu2022quantum}), there is a flexibility in choosing the mixer Hamiltonian $H_M$. The convergence of QAOA is ensured by the adiabatic theorem if $H_M$ satisfies certain conditions. For example, the assumptions outlined in the Perron-Frobenius theorem (Theorem \ref{PF}) are sufficient.

A commonly used mixer Hamiltonian consists of Pauli $X$-gates, $B = \sum\limits_{j=0}^{\ell-1} X_j$, where $\ell$ is the number of qubits required for the problem. However, this choice may not exploit problem-specific attributes. The choice of mixer Hamiltonian has been discussed in the literature. In \cite{HWORVB}, the authors introduced a quantum alternating operator ansatz to allow more general families of Hamiltonian operators. The mixers in that article are useful for optimization problems with hard  constraints that must always be satisfied (thus defining a feasible subspace of $W$) and soft constraints whose violation needs to be minimized.

In \cite{GPSK}, it was experimentally verified (via numerical simulations) that linear combinations of $X$- and $Y$-Pauli gates as mixers can outperform the standard low depth QAOA. More examples can be found in  \cite{BFL,GBOE,SW,ZLCMBE} and subsequent references.

Constructing an optimal mixer for QAOA Hamiltonian is crucial for enhancing the performance of QAOA in solving combinatorial optimization problems. Optimal mixers not only enforce hard constraints and align with the initial state for improved performance but also contribute to the universality and computational efficiency of QAOA, enabling the algorithm to exploit the structure of optimization problems for significant speed-ups and to adapt effectively to constrained problems. Here are a few of the examples:
\begin{enumerate}
    \item \textbf{Enforcing Hard Constraints:} The application of QAOA to problems with constraints presents a notable challenge, especially for near-term quantum resources. Utilizing $XY$ Hamiltonians as mixers has been shown to enforce hard constraints effectively. These mixers can be implemented without Trotter error in certain cases, and they demonstrate significant improvement in performance over traditional $X$ mixers in solving graph-coloring problems, a known challenge for classical algorithms \cite{wang2020xy}.

    \item \textbf{Alignment with Initial State:} The alignment between the initial state and the ground state of the mixing Hamiltonian has been observed to improve QAOA performance. This alignment, mimicking the adiabatic algorithm's requirements, has been particularly beneficial in constrained portfolio optimization, showcasing that an optimal mixer enhances results across different QAOA depths \cite{he2023alignment}.

    \item \textbf{Universality and Computational Efficiency:} The universality of QAOA with optimal mixers extends its applicability across a broader spectrum of problems. Optimal mixers contribute to the quantum computational universality, enabling the solution of complex optimization problems with high efficiency and precision. This universality underpins QAOA's potential in leveraging quantum computing for practical applications \cite{morales2020universality}.

    \item \textbf{Exploiting Problem Structure for Speed-Up:} Recent studies have provided numerical evidence that QAOA, with appropriately chosen mixers and phase separators, can significantly outperform classical unstructured search algorithms in finding approximate solutions to constrained optimization problems. This suggests that optimal mixers are key to leveraging the structure of optimization problems for computational speed-up \cite{golden2022speedup}.

    \item \textbf{Custom Mixers for Constrained Problems:} For constrained optimization problems, especially those involving network flows, custom mixers inspired by quantum electrodynamics (QED) have been shown to preserve flow constraints, leading to an exponential reduction in the configuration space to be explored. This adaptation results in higher quality approximate solutions, underscoring the importance of mixer customization \cite{zhang2020qed}.
\end{enumerate}

In this paper, we extend various investigation into tailoring the mixer Hamiltonian to accommodate groups of classical symmetries inherent in the objective function. In particular, our exploration builds upon our groundwork laid out in \cite{TSA}, where we detailed the construction of mixer Hamiltonians, along with their corresponding ground states, designed for cases where the group of classical symmetries includes the symmetric group $S_n$, encompassing permutations of string elements. While we  presented compelling arguments advocating for the adoption of such mixer Hamiltonians over classical counterpart, practical validation was hindered by the challenge of implementing the suggested matrices as concrete quantum circuits. Our current focus is on cases where the group of classical symmetries involves a different symmetric group, $S_d$, acting by simultaneous permutation of all factors in $\mathbb{D}^n$:
\[
\sigma(d_1,d_2,\ldots,d_n):=(\sigma(d_1),\sigma(d_2),\ldots,\sigma(d_n)).
\]
Considering such cases offers two significant advantages:

\begin{itemize}
\item Many optimization problems exhibit these symmetries (e.g., several versions of graph coloring and partitioning).
\item We can construct a mixer Hamiltonian that commutes with the action of $S_d$ on $W$, which can be easily implemented as a composition of basic quantum gates.
\end{itemize}

The paper is structured as follows for clear and systematic exposition. Section $2$ offers an overview of the main results to orient the reader. In Section $3$, a concise review of Quantum Approximate Optimization Algorithm fundamentals relevant to this study is provided.

Section $4$ presents formulations of the main results and delineates properties concerning the newly proposed Hamiltonian. The subsequent sections, $5$ and $6$,  respectively, explore the classical optimization problems under consideration and provide simulation results for three QAOA versions: one utilizing the classical mixer and the others employing the newly proposed mixers.

In Section $8$, the impossibility of tailoring a mixer Hamiltonian that satisfies the Perron-Frobenius theorem within the context of warm-start QAOA is discussed. Finally, the Appendix offers a conceptual overview of the construction process and provides rigorous verification of the claims made throughout the paper.

\section{Main results} We present a systematic approach to constructing a mixer Hamiltonian for QAOA that aligns with the symmetries inherent in the objective function of the classical optimization problem being addressed. Specifically, our approach focuses on identifying an operator that commutes with the action of the group of symmetries on the Hilbert space for QAOA. Additionally, this operator fulfills the necessary technical requirements to function effectively as a mixer Hamiltonian.

In this work, we provide a method for constructing such an operator tailored to the aforementioned group $S_d$. This group naturally emerges as a set of symmetries in numerous combinatorial optimization problems. We rigorously validate the required properties for the proposed mixer Hamiltonian, $H_M$, offering both a concrete formula and a corresponding quantum circuit for its implementation. In addition, we show that the classical mixer $B$ commutes only with a subgroup of $S_d$ of order $2^{\ell}\cdot\ell!$ (in case $d=2^\ell$ is a power of two), which is significantly smaller than $d!$, the order of $S_d$.  

Furthermore, we explore the cyclic subgroup $\mathbb{Z}_d$ within $S_d$, generated by the element $g:=(23 \ldots n1)$. This generator cyclically shifts  $1$ to $2$, $2$ to $3$, and $n$ to $1$. Subsequently, we construct an operator $H_\chi$ whose action on $W$ commutes with $\mathbb{Z}_d$, and has the state $|\psi\rangle:= |\underbrace{-+\ldots +}_{n\ell} \rangle$ as its unique ground state. Notably, the Hilbert space $W =\underset{j=0}{\overset{d-1}{\bigoplus}} W_j$ decomposes into a direct sum of equidimensional vector spaces decomposes into a direct sum of equidimensional vector spaces with respect to the $\mathbb{Z}_d$-action, with $|\psi\rangle$ situated in the subspace $W_{d/2}$. Moreover, the images of $|\psi\rangle$ during the execution of the QAOA with $H_\chi$ in place of the mixer Hamiltonian remain within this subspace until the final projection (see Appendix $A$ for precise results). To the best of our knowledge, this is the first example of a QAOA algorithm realized entirely (with the exception of the final measurement) within a nontrivial representation of a symmetry group of the objective function. 

We proceed by evaluating the effectiveness of simulations of three QAOA variants employing distinct mixer Hamiltonians: the conventional $B=\sum X_i$ and the newly proposed $H_M$ and $H_\chi$, applied to the edge coloring  and graph partitioning problems across a range of graphs. Both algorithms are configured iteratively with a depth parameter of $p=9$ for edge coloring and $p=7$ for graph partitioning, respectively. Through $50$ or more independent trials for each scenario, we observe statistically significant differences in mean values  at the $1.5\%$ significance level, with the new variant consistently demonstrating lower means. Moreover, we note considerably lower median and minimal values in the experiments utilizing the newly introduced mixer Hamiltonians compared to the classical one (see Section $6$ for details). 

Finally, we address an intriguing observation regarding the subpar performance of warm-start QAOA variants—a phenomenon recently documented in the literature. Warm-start strategies involve initiating QAOA from a promising classical solution generated by a classical algorithm, with the aim of further refining it through quantum optimization. While this approach has garnered a lot of attention in recent studies \cite{sridhar2023adapt,egger2021warm,okada2024systematic}, our investigation sheds light on its fundamental limitations.

In a recent study by  \cite{CFGRT}, extensive numerical experiments across a range of problem sizes and depths uncovered a significant finding. Notably, when QAOA initializes from a single warm-start string, it demonstrates minimal progress. We provide a conceptual elucidation for this observation. Specifically, we identify the absence of an operator satisfying the assumptions of the Perron-Frobenius theorem while also possessing a superposition of classical states with identical objective function value  as its ground states. This absence undermines the convergence guarantee of any warm-start QAOA variant to an optimal solution, even in the limit as the depth parameter approaches infinity ($p\to\infty$). 

Consequently, the convergence of warm-start QAOA variants to an optimal solution hinges entirely on the classical optimizer's ability to avoid being trapped in parameter sets, leading to local extrema of the objective function and raising a major question to a variety of warm-start heuristics that claim observing quantum advantage, namely, "is the advantage indeed quantum?"

\section{Overview of QAOA}

Let $\mathbb{D}^n:=\{0,1,\hdots,d-1\}^n$ be the set of $n$-element strings and $\mathcal{S}$ the group of  permutations of these $d^n$ elements. A classical optimization problem can be formulated as follows: given a function  $F:\mathbb{D}^n \rightarrow \mathbb{R}$, find the elements in $\mathbb{D}^n$ on which it attains min (max) values. If a permutation $g\in \mathcal{S}$ is undetectable by $F$, i.e. $F(g(x))=F(x)$ for any $x\in \mathbb{D}^n$, then $g$ is symmetry of $F$. Such elements form a subgroup $G\subset\mathcal{S}$ and $F$ is invariant with respect to this subgroup.

One of the widely employed algorithms for tackling the quantum version of the optimization problem  is the Quantum Approximate Optimization Algorithm, introduced in \cite{QAOA}. In the QAOA framework, the Hamiltonian $H_F$ is commonly referred to as the \textit{problem Hamiltonian} and is denoted by $H_P$ (as per Farhi's et al. paper \cite{QAOA}). We will adopt this notation consistently.

Central to QAOA is the mixer Hamiltonian $H_M$, characterized by a distinct lowest energy state $|\xi\rangle \in W$ and adherence to the requirements of the Perron-Frobenius theorem (refer to Theorem \ref{PF}). The core idea behind the QAOA algorithm lies in iteratively transforming the mixer Hamiltonian $H_M$ into the problem Hamiltonian. This process ensures that the image of the lowest-energy vector from the preceding step becomes the lowest-energy vector in the subsequent one.

The algorithm initiates by preparing the state $|\xi\rangle$, the ground state for the mixer Hamiltonian $H_M$, and then proceeds with multiple alternating applications of (certain exponents of) the problem and mixer Hamiltonians. The number of iterations is conventionally denoted by $p$ (also known as QAOA depth), and we use $\mathfrak{Q}_p$ to express the entire composition of operators
\begin{equation}\label{qaoa-chain}
\mathfrak{Q}_p:=e^{-i\beta_1 H_M}e^{-i\gamma_1 H_P}\hdots e^{-i\beta_p H_M}e^{-i\gamma_p H_P}.
\end{equation}

The final step of QAOA involves performing a measurement of the state obtained after applying $\mathfrak{Q}_p$ in the standard basis. For an in-depth description of the algorithm, we direct the reader to Section $2$ and the references therein.

While the Hamiltonian $H_P$, representing the objective function, is uniquely determined by the classical problem, there is some flexibility in choosing the pair of mixer Hamiltonian and initial state. The convergence of QAOA to a classical state representing an element on which $F$ attains a minimum value is guaranteed by the adiabatic theorem, provided the mixer Hamiltonian satisfies the conditions of the Perron-Frobenius theorem (see below and Theorem $8.4.4$ in \cite{HJ}) and the initial state is the ground state for it.

\begin{thm} (Perron-Frobenius). Let $M=(m_{ij})\in \mbox{Mat}_n(\mathbb{R})$ be an irreducible matrix with $m_{ij}\geq 0$.
\begin{itemize}
\item Then there is a positive real number $r$, such that $r$ is an eigenvalue of $M$ and any other eigenvalue $\lambda$ (possibly complex) has $\operatorname{Re}(\lambda)<r$. 
    \item Moreover, there exists a unique real vector  $v=(v_1,v_2,\ldots,v_n)$ such that $M(v)=rv$ and $v_1+v_2+\ldots+v_n=1$. This vector is positive, i.e. all $v_i$ are strictly greater than $0$. 
\end{itemize}

\label{PF}
\end{thm}

The standard and most common choice of mixer Hamiltonian involves Pauli $X$-gates and is given by $B=\sum\limits_{0\leq j \leq \ell-1} X_j$, where $\ell$ is the number of qubits needed for the (re)formulation of the original problem. The corresponding ground state is $|\xi\rangle=|+\rangle^{\otimes l}$. While this choice offers certain advantages, it does not consider any specific attributes of a given problem, in particular, the group of symmetries $G$.

\section{Symmetries of the Mixers}


In this section, we offer a broad, high-level overview of our approach to selecting the mixer Hamiltonian based on symmetries inherent in the objective function of the optimization problem being addressed. A more comprehensive and conceptual discussion is deferred to the appendix.

When determining the symmetries, it is natural to start by considering the group $\mathcal{S}$ consisting of all permutations of the elements within the set of all $d$-element strings $\mathbb{D}^n$. This action naturally extends to an action on classical states, and by linearity, to the vector space $W$ associated with $\mathbb{D}^n$. The group of classical symmetries for an optimization problem forms a subgroup $G$ comprising elements $g \in \mathcal{S}$ that remain 'undetectable' by $F$, meaning that $F(g(x)) = F(x)$ for any $x \in \mathbb{D}^n$. It is straightforward to observe that elements in this subgroup commute with the action of the problem Hamiltonian (representing $F$) on $W$. It is natural to seek a mixer Hamiltonian that satisfies the necessary technical requirements of the Perron-Frobenius theorem (see Theorem \ref{PF}), ensuring convergence as $p \rightarrow \infty$ and commuting with the largest subgroup of $G$, ideally encompassing the entire group $G$. 
Given that the latter condition implies that the corresponding unitary operator $\mathfrak{Q}_p$, which is the product of $p$ alternating applications of mixer and problem Hamiltonian operators, commutes with $G$, it is natural to refer to the corresponding QAOA as $G$-equivariant.

Within $\mathcal{S}$, there exists a subgroup $S_d = \text{\textit{Perm}}(\mathbb{D})$, comprising permutations of elements within a single copy of the symbol set $\mathbb{D}$. This subgroup acts by simultaneously permuting elements of $\mathbb{D}^n$ in the same manner across all copies:

$$g(d_1,d_2,\ldots,d_n):=(g\cdot d_1,g\cdot d_2,\ldots,g\cdot d_n). $$

In many optimization problems (as discussed in the following sections), the objective function contains $S_d$ as a subgroup of its classical symmetries, i.e., $S_d \subseteq G$.

For simplicity of exposition and to facilitate future practical implementations, we will assume that the number of elements, denoted by $d$, is a power of two, i.e., $d = 2^{\ell}$. In this case, two subgroups of $S_d$ will play a fundamental role in our discussion. To describe them, it is convenient to consider the set $\mathbb{D}$ as a union of $\ell$ bits.

The group $K_\ell := \underbrace{\mathbb{Z}_2 \times \ldots \times \mathbb{Z}_2}_{\ell}$, which is a subgroup of $\mathcal{S}$, represents the bit flips for each of these bits. Meanwhile, $S_{\ell} \subset S_{d}$ is the subgroup responsible for permuting the bits.

In the appendix, we elaborate on the construction (and the reasoning behind it) of a mixer Hamiltonian $H_M$, whose action on $W$ (the vector space corresponding to $\mathbb{D}^n$) commutes with the action of the entire group $S_d$. Importantly, $H_M$ satisfies the assumptions of the Perron-Frobenius theorem. 

 Similar to the classical mixer Hamiltonian $B$, the operator $H_M$ has a uniform superposition of all classical states \[|\xi\rangle = \frac{1}{2^{n\ell }}H^{\otimes {n\ell} }(|\underbrace{00\ldots 0}_{n\ell} \rangle) = |\underbrace{++\ldots +}_{n\ell} \rangle\] as its unique ground state. However, we also highlight a significant difference between the two operators, $B$ and $H_M$. Specifically, the action of the classical mixer Hamiltonian $B$ only commutes with a smaller subgroup, which is the semidirect product of the groups $K_\ell$ and $S_\ell$, and has an order of $2^\ell \cdot \ell!$. This is notably less than the order of $S_d$, which is $d! = 2^\ell !$, as demonstrated in Proposition \ref{Centralizers} and the subsequent Corollary \ref{GroupCor}. 

We proceed by examining the cyclic subgroup $\mathbb{Z}_d$ within $S_d$, generated by the element $g:=(23 \ldots n1)$, which cyclically shifts the elements from $1$ to $2$, $2$ to $3$, and $n$ to $1$. We then construct an operator $H_\chi$ whose action on $W$ commutes with $\mathbb{Z}_d$ and has the state $|\psi\rangle:= \frac{1}{2^{n\ell }}H^{\otimes n\ell}(|\underbrace{10\ldots 0}_{n\ell}\rangle) = |\underbrace{-+\ldots +}_{n\ell} \rangle$ as its unique ground state.

\begin{rmk}
The ambient Hilbert space $W$ admits a decomposition into a direct sum of subspaces:

\[ W =\underset{j=0}{\overset{d-1}{\bigoplus}} W_j \]
according to the $\mathbb{Z}_d$-action. It is interesting to note that the state vector $|\xi\rangle$ resides in $W_0$, while $|\psi\rangle$ is located in $W_{d/2}$. Moreover, the images of these vectors during the execution of their respective QAOAs remain within these subspaces prior to the final projection (see Remark \ref{subspacermk} for a precise statement).

\end{rmk}

Let us reiterate that we defer the verification of the existence of the operators $H_M$ and $H_\chi$ satisfying the aforementioned properties to the appendix. Instead, our focus in the subsequent sections will be on demonstrating its practical advantages over the classical mixer.

\section{Outline of the Two Problems}

In this section, we elucidate two significant classical optimization problems and their reformulations within the framework of QAOA. These problems  find numerous applications across various domains \cite{bulucc2016recent,jensen2011graph}.

\subsection{Problem 1: Edge Coloring}

One class of optimization problems with objective function having the aforementioned group of symmetries, $S_d$, is coloring of the vertices or edges of a graph in $d$ colors. 

\begin{defn}
A \textbf{vertex coloring} of a graph $\Gamma = (V, E)$ is a map $\widetilde{C} : E \rightarrow \mathfrak{C}$, where $\mathfrak{C}$ is a set of colors with $|\mathfrak{C}|=d$. A coloring $\widetilde{C}$ is called \textbf{proper} if  $\widetilde{C}(v_1)\neq \widetilde{C}(v_2)$ for any two adjacent vertices $v_1, v_2 \in V$.

Similarly, an \textbf{edge coloring} of a graph $\Gamma = (V, E)$ is a map $C : E \rightarrow \mathfrak{C}$. A coloring $C$ is called \textbf{proper} if  $C(e)\neq C(f)$ for any two adjacent edges $e, f \in E$.
\end{defn}

To represent $k$ colors, we employ $\ell=\ell og_2(d)$ bits through the following encoding:

\begin{align*}
& color_0\longleftrightarrow 0\ldots 00\\
& color_1\longleftrightarrow 0\ldots 01\\
& \ldots\\
\end{align*}

In this section we focus on the edge coloring. Each edge $e\in E$ is assigned $\ell$ bits $e_0, e_1,\ldots, e_{\ell-1}$ whose values uniquely determine the color of the edge. The characteristic function of a color \(C \in \mathfrak{C}\) is defined as follows:

\[ \chi_c(c') :=
\begin{cases} 
1, & \text{if } c'_i \equiv c_i \text{ for all } i \in \{1, \ldots, \ell\} \\
0, & \text{otherwise}
\end{cases}
\]

This function, denoted as \(\chi_c(C(e))\), is explicitly given by

\[ \chi_c(C(e)) = \prod\limits_{i=1}^{\ell}((1-c_i)e_i+c_i(1-e_i)) \]

This defining property ensures that the characteristic function equals $1$ on the specific color \(C\) and $0$ on all other colors. 
We define the objective function
\[ F_{\Gamma}(C):=\sum\limits_{e\bullet f}\sum\limits_{c\in \mathfrak{C}}\chi_c(C(e))\chi_c(C(f)), \]
where the notation $e \bullet f$ represents adjacent edges. This function calculates the number of adjacent edges with the same color.

\begin{rmk}
A coloring $C$ is proper if and only if $F_{\Gamma}(C)=0$.
\end{rmk}

It is evident that the action of the group $S_d$, permuting the colors, preserves the values of the objective function:
\[ F_{\Gamma}(\sigma^{-1} (C))=F_{\Gamma}(C) \quad \forall \sigma\in S_d, C\in \mathfrak{C}. \]

\begin{defn}
The \textbf{chromatic index} $\chi_\Gamma$ of a graph $\Gamma$ is the minimum number of colors needed for a proper coloring of $\Gamma$.
\end{defn}

The following result was proved in \cite{VIZ}.
\begin{thm}
Let  $\Gamma$ be a simple undirected graph with maximum degree $\triangle(\Gamma)$. Then $\triangle(\Gamma)\leq\chi(G)\leq \triangle(\Gamma)+1$.
\label{Vizing}
\end{thm}

\begin{defn}
Graphs that can be colored with $\triangle(\Gamma)$ are called \textbf{class one}
graphs. Graphs that require at least $\triangle(\Gamma)+1$ colors are called \textbf{class two} graphs.
\end{defn}

In order to resolve the dichotomy in Theorem \ref{Vizing}, whether the minimal proper coloring of edges involves $k$ or $k+1$ colors, it suffices to find out if a proper $k$ coloring exists.

The operator representing the characteristic function $\chi_c$ is given by

\[
\widetilde{\chi}_c(e) := 
\begin{cases} 
    e, & e_i \equiv c_i ~ \forall i \in \{1,\ldots, \ell\} \\
    0, & \text{otherwise}
\end{cases}
\]

and is expressed as 
\[
\widetilde{\chi}_c(e) = \frac{1}{2^\ell}\bigotimes\limits_{i=1}^{\ell}(\mathds{1} + (-1)^{c_i}Z_{e,i}),
\]

In the case $\ell=2$, this expression simplifies to 
\[
Z_{e,0}Z_{f,0}Z_{e,1}Z_{f,1} + Z_{e,0}Z_{f,0} + Z_{e,1}Z_{f,1} + \lambda\mathds{1}.
\]

Meanwhile, the problem Hamiltonian representing  $F_{\Gamma}$ is 
\[
H_P = \sum\limits_{e\bullet f}\sum\limits_{c\in \mathfrak{C}}\widetilde{\chi}_c(e)\widetilde{\chi}_c(f).
\]

The building blocks for the quantum circuit representing the exponent of the latter operator,

\[
e^{-i\beta H_P} = \prod\limits_{e\bullet f}e^{-i\beta(Z_{e,0}Z_{f,0}Z_{e,1}Z_{f,1} + Z_{e,0}Z_{f,0} + Z_{e,1}Z_{f,1})} = \prod\limits_{e\bullet f}(e^{-i\beta Z_{e,0}Z_{f,0}Z_{e,1}Z_{f,1}}e^{-i\beta Z_{e,0}Z_{f,0}}e^{-i\beta Z_{e,1}Z_{f,1}})
\]

are presented in Figure \ref{ZZZZcircuit} below.

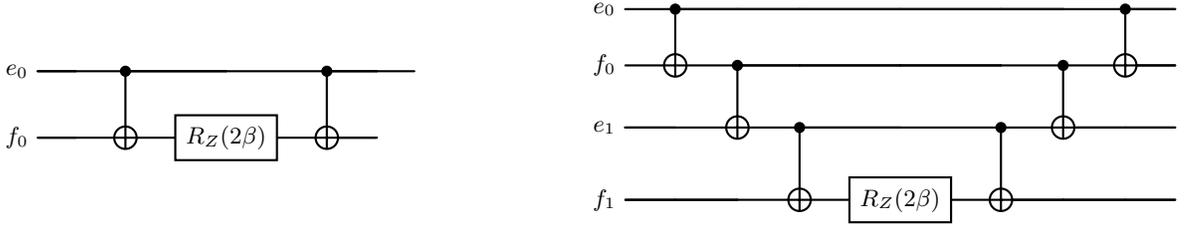
\begin{figure}[htbp!]
  \begin{minipage}{0.5\linewidth}
    \centering
    \begin{quantikz}
      e_0\hspace{0.05in} & \qw & \ctrl{1} & \qw & \ctrl{1} & \qw &  \\
      f_0\hspace{0.05in} & \qw	 & \targ{} & \gate{R_{Z}(2\beta)} & \targ{} & \qw\\
    \end{quantikz}
  \end{minipage}%
  \begin{minipage}{0.5\linewidth}
    \centering
    \begin{quantikz}
      e_0\hspace{0.05in} & \ctrl{1}	& \qw  & \qw & \qw &  \qw & \qw & \ctrl{1} & \qw \\
      f_0\hspace{0.05in} & \targ{}	& \ctrl{1}   & \qw & \qw & \qw & \ctrl{1} & \targ{} & \qw\\
      e_1\hspace{0.05in} & \qw	& \targ{}  & \ctrl{1} & \qw & \ctrl{1} & \targ{} & \qw & \qw \\
      f_1\hspace{0.05in} & \qw	& \qw & \targ{} & \gate{R_{Z}(2\beta)} & \targ{} & \qw & \qw & \qw\\
    \end{quantikz}
  \end{minipage}
   \caption{Quantum circuits for $e^{-i\beta Z_{e,0}Z_{f,0}}$ and $e^{-i\beta Z_{e,0}Z_{f,0}Z_{e,1}Z_{f,1}}$}
    \label{ZZZZcircuit}
\end{figure}

\subsection{Problem 2: graph partitioning}

The second optimization problem explored in this paper is the balanced graph partitioning problem. This problem appears in numerous applications \cite{bulucc2016recent} and has been a subject of several investigations in QAOA and other frameworks \cite{ushijima2021multilevel,shaydulin2019multistart}. 
Given a graph $\Gamma$ and a fixed integer $k$ that divides the number of vertices in $\Gamma$, the objective is to find a partition of the vertices: $V = V_0 \sqcup V_1 \sqcup \hdots \sqcup V_{k-1}$ into $k$ disjoint subsets of equal cardinality that minimizes the total number of cut edges. A cut edge is defined as an edge with endpoints in different subsets. The requirement of exact equality of sizes of $V_i$'s for all $i$ is often referred to as perfectly balanced graph partitioning.

This problem bears some resemblance to the vertex coloring problem for $k$ colors. Specifically, we can refer to vertices in subset $V_i$ as colored with the $i$-th color. However, unlike the coloring problem where we aim to minimize the number of adjacent vertices with the same color, here we seek to maximize this number. Additionally, we must account for the restriction on the cardinalities of the $V_i$'s. 

We will examine examples with $k=4$ and the number of vertices in the graph being a multiple of $4$. As before, we encode the $4$ colors using $2$ bits. We define the objective function $F(C)$ as follows: 
\[ F(C) = -\sum\limits_{v-v'} \sum\limits_{c \in \mathfrak{C}} \chi_c(C(v)) \chi_c(C(v'))+ \left(2E\sum\limits_{v\in V} (v_0-0.5)\right)^2 + \left(2E\sum\limits_{v\in V} (1-v_0)(v_1-0.5)\right)^2 + \left(2E\sum\limits_{v\in V} v_0(v_1-0.5)\right)^2, \]
where the notation $v-v'$ is used for adjacent vertices. The first sum evaluates the number of pairs of adjacent vertices belonging to different subsets of the partition, while the remaining three ensure that $|V_0|=|V_1|=|V_2|=|V_3|=\frac{|V|}{4}$. Specifically, $\left(\sum\limits_{v\in V} (v_0-0.5)\right)^2$ equals zero if and only if the numbers of vertices with the first color bit equal to $0$ and $1$ coincide; otherwise, it is positive. Similarly, $\left(\sum\limits_{v\in V} (1-v_0)(v_1-0.5)\right)^2$ and $\left(\sum\limits_{v\in V} v_0(v_1-0.5)\right)^2$ equal zero if and only if the numbers of vertices with the second color bit equal to $0$ and $1$ coincide, respectively, for the first color bit being fixed at $0$ and $1$.

The corresponding problem Hamiltonian is given by:
\[ H_P = -\sum\limits_{v-v'} \sum\limits_{c \in \mathfrak{C}} \widetilde{\chi}_v(c) \widetilde{\chi}_{v'}(c) + E\left(\sum\limits_{v\in V} Z_{v,0}\right)^2 + E\left(\sum\limits_{v\in V} (1-Z_{v,0})Z_{v,1}\right)^2 + E\left(\sum\limits_{v\in V} (1-Z_{v,0})Z_{v,1}\right)^2 \]

and is equivalent to 

\[-\sum\limits_{v-v'} \sum\limits_{c \in \mathfrak{C}} \widetilde{\chi}_v(c) \widetilde{\chi}_{v'}(c) + 2E\sum\limits_{v,v'\in V} (Z_{v,0}Z_{v',0}+Z_{v,1}Z_{v',1})+ 2E\sum\limits_{v,v'\in V} Z_{v,0}Z_{v',0}Z_{v,1}Z_{v',1}.\]

\section{The Duel: equivariant $H_M, H_\chi$ Vs classical $B$}

In this section, we contrast the performance of QAOA algorithms using different mixer Hamiltonians: the classical one, $B=\sum X_i$, and the newly introduced equivariant $H_M$ and $H_\chi$. We analyze their effectiveness on the problems discussed in the preceding section, primarily comparing $H_M$ and $H_\chi$ with $H_B$.
We implement the  algorithms iteratively. The algorithms begin by establishing the initial state:

\[
|\xi\rangle = \frac{1}{2^{n\ell }}H^{\otimes {n\ell} }(|\underbrace{00\ldots 0}_{n\ell} \rangle) = |\underbrace{++\ldots +}_{n\ell} \rangle
\]

for \(H_M\) and \(H_B\), or 

\[
|\psi\rangle = \frac{1}{2^{n\ell }}H^{\otimes n\ell}(|\underbrace{10\ldots 0}_{n\ell}\rangle) = |\underbrace{-+\ldots +}_{n\ell} \rangle
\]
for \(H_\chi\). The initial pair of parameters \((\beta_1,\gamma_1)\) is randomly selected from the uniform distribution on the set \([0,0.25\pi]\times[0,2\pi]\). Subsequently, the algorithm iterates through runs: after completing the $p=1$ run, optimal values $(\beta^*_1,\gamma^*_1)$ are determined with the aid of a classical optimizer. The subsequent QAOA run is then executed with starting parameters $(\beta^*_1,\gamma^*_1,0,0)$ for $p=2$, and this process continues iteratively. The objective of the classical optimizer is to minimize the \textit{energy}, which is defined as the average value of the objective function on the states output by the algorithm over multiple runs:

\begin{equation}
\mathcal{E}_p:=\dfrac{\sum\limits_{i=1}^m F_{\Gamma}(\mathfrak{Q}_p(|s_i\rangle))}{m}.\label{eq:energy}
\end{equation}

\begin{rmk}
   In case of the edge coloring problem, if the energy  $\mathcal{E}_p < 1$, it implies that at least one of the obtained values $F_{\Gamma}(\mathfrak{Q}_p(|s_i\rangle))$ is zero. Consequently, the corresponding coloring is proper, indicating that $\Gamma$ is a class one graph.
\end{rmk}

On each successive step, the starting parameters consist of the values converged by the classical optimizer on the preceding step, complemented by two zeros for the additional angles that did not appear in the previous step. This deliberate choice ensures that the energy values $\mathcal{E}_1,\mathcal{E}_2,\ldots$ obtained in subsequent steps are nonincreasing, as outlined in \cite{QAOA}.  The algorithm's depth for the edge coloring problem was set at $p=9$ and for the graph partitioning problem at $p=7$. We conducted multiple independent simulations, ranging from $50$ to $56$, for various graphs (see  Figures \ref{graphscoloring} and \ref{graphspartitioning}) using the qiskit codes available at \href{https://github.com/BorisTsv/QAOA-Mixer-Hamiltonians-for-Optimization-Problems-with-S_d-Symmetries/tree/main}{this link}. The main characteristics of the outcomes are summarized in Table \ref{Table1}, while the histograms displaying the average $\mathcal{E}_p$-values across sample runs of the equivariant algorithms, as compared to the classical one for each graph, are depicted in Figures \ref{Histograms1} and \ref{Histograms2}.

Based on  Table \ref{Table2}, which presents the Student's $t$-test values for testing the hypothesis that the means of energy values for the two algorithms are equal, we reject this hypothesis at a significance level of $\alpha=1.5\%$ for all graphs analyzed (with the exception of graph $ \Gamma_2$ for the mixer $H_\chi$). This indicates a statistically significant difference in the energy values between the algorithms across all examined graphs. Furthermore, we consistently observe lower median and minimal  energy values (columns $3$ and $4$) for the newly proposed mixers.

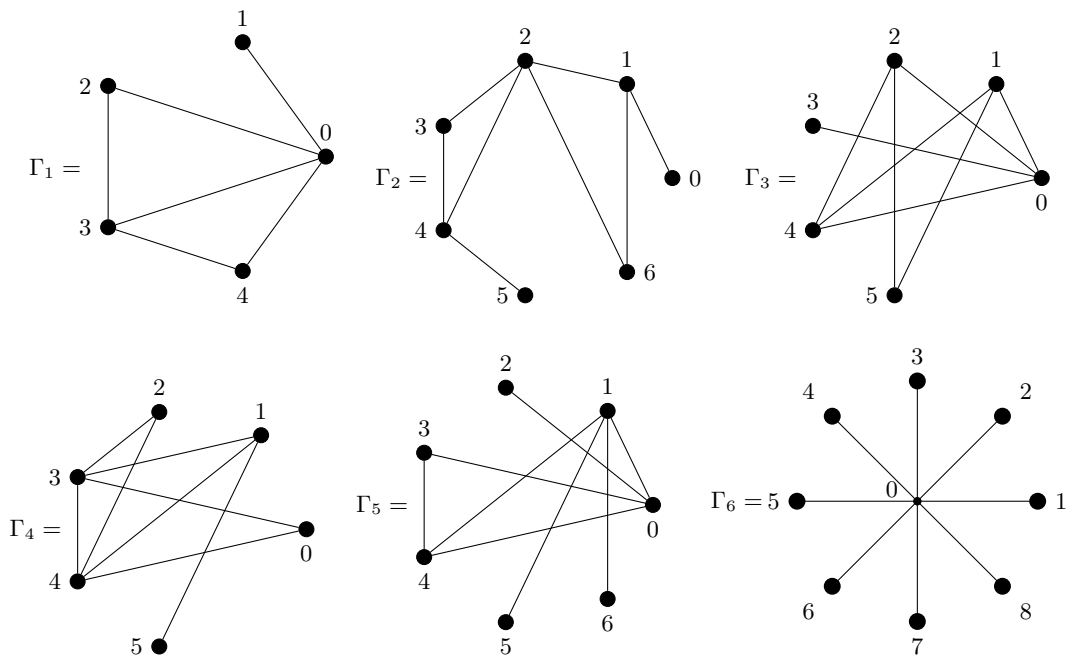
\begin{figure}[htbp]
\begin{center}
\begin{tikzpicture}[scale=0.8]
  \foreach \i in {0,1} {
    \node[circle, draw, fill=black, inner sep=2pt, label={above:$\i$}] (\i) at (\i*360/5:2) {};
  }
  \foreach \i in {2,3} {
    \node[circle, draw, fill=black, inner sep=2pt, label={left:$\i$}] (\i) at (\i*360/5:2) {};
  }
  \foreach \i in {4} {
    \node[circle, draw, fill=black, inner sep=2pt, label={below:$\i$}] (\i) at (\i*360/5:2) {};
  }

  \foreach \u/\v in {0/1, 0/2, 0/3, 0/4, 2/3, 3/4} {
    \draw (\u) -- (\v);
  }
  \node at (-2.5,-0.2) {$\Gamma_1=$};
\end{tikzpicture}\hspace{0.1in}
\begin{tikzpicture}[scale=0.8]
  \foreach \i/\pos in {0/right, 1/above, 2/above, 3/left, 4/left, 5/left, 6/right}
    \node[circle, draw, fill=black, inner sep=2pt, label={\pos:$\i$}] (\i) at (\i*360/7:2) {};
  
  \foreach \x/\y in {0/1, 1/2, 2/3, 3/4, 4/5, 1/6, 2/6, 2/4}
    \draw (\x) -- (\y);
    \node at (-2.5,0) {$\Gamma_2=$};
\end{tikzpicture}\hspace{0.1in}
\begin{tikzpicture}[scale=0.8]
  \foreach \i/\pos in {0/below, 1/above, 2/above, 3/above, 4/left, 5/left}
    \node[circle, draw, fill=black, inner sep=2pt, label={\pos:$\i$}] (\i) at (\i*360/7:2) {};
  
  \foreach \x/\y in {0/1, 0/2, 0/3, 0/4, 1/5, 1/4, 2/5, 2/4}
    \draw (\x) -- (\y);
\node at (-2.5,0) {$\Gamma_3=$};
\end{tikzpicture}
\end{center}

\begin{center}

\begin{tikzpicture}[scale=0.8]
  \foreach \i/\pos in {0/below, 1/above, 2/above, 3/left, 4/left, 5/left}
    \node[circle, draw, fill=black, inner sep=2pt, label={\pos:$\i$}] (\i) at (\i*360/7:2) {};
  
  \foreach \x/\y in {3/1, 3/2, 0/3, 0/4, 1/5, 1/4, 3/4, 2/4}
    \draw (\x) -- (\y);
        \node at (-2.5,0) {$\Gamma_4=$};
\end{tikzpicture}\hspace{0.1in}
  \begin{tikzpicture}[scale=0.8]
  \foreach \i in {1,2,3}
    \node[circle, draw, fill=black, inner sep=2pt, label={above:$\i$}] (\i) at (\i*360/7:2) {};
  
  \foreach \i/\pos in {0/below, 4/below, 5/below, 6/below}
    \node[circle, draw, fill=black, inner sep=2pt, label={below:$\i$}] (\i) at (\i*360/7:2) {};
  
  \foreach \x/\y in {0/1, 0/2, 0/3, 0/4, 1/5, 1/6, 1/4, 3/4}
    \draw (\x) -- (\y);
    \node at (-2.5,0) {$\Gamma_5=$};
\end{tikzpicture}  \hspace{0.1in}
\begin{tikzpicture}[scale=0.8]
    \foreach \k/\ang in {1/0, 2/45, 3/90, 4/135, 5/180, 6/225, 7/270, 8/315} {
    \fill (\ang:2) circle (4pt) node[label={\ang:$\k$}] {};
    \draw (0,0) -- (\ang:2);
  }
  \fill (0,0) circle (2pt);
  \node at (-0.42,0.2) {$0$};
    \node at (-3,0) {$\Gamma_6=$};
\end{tikzpicture}
 \caption{Graphs $\Gamma_1,\Gamma_2,\Gamma_3,\Gamma_4,\Gamma_5$ and $\Gamma_6$ considered for edge coloring problem}
    \label{graphscoloring}
\end{center}
\end{figure}

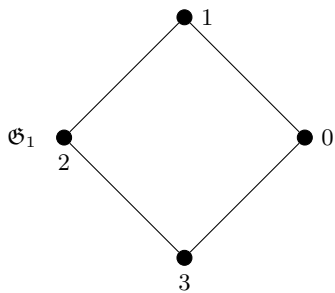
\begin{figure}[htbp]
\begin{center}
\begin{tikzpicture}[scale=0.8]
  \foreach \i in {0,1} {
    \node[circle, draw, fill=black, inner sep=2pt, label={right:$\i$}] (\i) at (\i*90:2) {};
  }
  \foreach \i in {2,3} {
    \node[circle, draw, fill=black, inner sep=2pt, label={below:$\i$}] (\i) at (\i*90:2) {};
  }

  \foreach \u/\v in {0/1, 1/2, 2/3, 3/0} {
    \draw (\u) -- (\v);
  }
  \node at (-2.7,0) {$\mathfrak{G}_1$};
\end{tikzpicture}
\end{center}
\caption{Graph $\mathfrak{G}$  considered for graph partitioning problem}
    \label{graphspartitioning}
\end{figure}
\newpage

\begin{table}[htbp]
  \centering
  \resizebox{0.45\textwidth}{!}{\begin{tabular}{|c|c|c|c|c|}
    \hline
    Graph & Mean & Median & Min & $\mathcal{E}_9<1$  \\
    \hline
    $\Gamma_1, B$ & 0.726 & 0.7056 & 0.3584&  41/50  \\
     \hline
    $\Gamma_1, H_M$ & 0.5692 & 0.4673 & 0.1923&47/50 \\
    \hline
    $\Gamma_1, H_\chi$ & 0.5726 & 0.5142 & 0.1621 &47/50 \\
    \hline
    $\Gamma_2, B$ & 0.9696  &0.9316  & 0.4814 & 33/56 \\
   \hline
   $\Gamma_2, H_M$  &0.7437  &0.7388  &0.3691  & 51/56  \\
    \hline
    $\Gamma_2, H_\chi$  & 0.8688  &0.7148  &0.3964 &47/56  \\
    \hline
    $\Gamma_3, B$ & 1.2495 & 1.2417 & 0.6533 & 11/56  \\
  \hline
    $\Gamma_3, H_M$ & 0.9344 & 0.8857 & 0.3691 & 35/56 \\
    \hline
     $\Gamma_3, H_\chi$ &0.7334 &0.6763 &0.2598 &50/56  \\
    \hline
      $\Gamma_4, B$ & 1.4857 & 1.5313 & 0.7382 & 6/56\\
     \hline
    $\Gamma_4, H_M$ & 1.1959 & 1.1074 & 0.5117 & 21/56\\
    \hline
    $\Gamma_4, H_\chi$ & 1.2415 & 1.1489 & 0.4395 &20/56 \\
    \hline
    $\Gamma_5, B$ & 1.3469 & 1.3066 & 0.6162 & 14/50\\
     \hline
    $\Gamma_5, H_M$ & 0.9149 & 0.9507 & 0.3516 & 30/50\\
    \hline
    $\Gamma_5, H_{\chi}$ & 0.94123 & 0.9375 & 0.2939& 27/50\\
    \hline
     $\Gamma_6, B$ & 0.8726 & 0.8569 & 0.502&23/28  \\
\hline
    $\Gamma_6, H_M$ & 0.5227 & 0.5073 & 0.17& 28/28 \\
    \hline

  \end{tabular}}
  \caption{QAOA performance comparison for edge coloring problem}
  \label{Table1}
\end{table}

\begin{table}[htbp]
  \centering
  \resizebox{0.4\textwidth}{!}{\begin{tabular}{|c|c|c|c|}
    \hline
    Graph & Mean & Median & Min   \\
    \hline
    $\mathfrak{G}, B$ & 10.08135 & 10.34229 & 6.24414   \\
     \hline
    $\mathfrak{G}, H_M$ & 8.05236 & 8.11035 & 4.47656 \\
    \hline
    $\mathfrak{G}, H_\chi$ &9.039  &8.888  &4.94434   \\
    \hline 
  \end{tabular}}
  \caption{QAOA performance comparison for graph partitioning problem}
  \label{Table1'}
\end{table}

\begin{table}[htbp]
  \centering
  \resizebox{0.45\textwidth}{!}{\begin{tabular}{|c|c|c|}
    \hline
    Graph/Mixer & $H_M$ & $H_\chi$ \\
    \hline
    $\Gamma_1$ & $0.01252$ & $0.007$ \\
     \hline
    $\Gamma_2$ & $3.054\cdot10^{-7}$&$0.237$\\
    \hline
    $\Gamma_3$ & $1.9807\cdot10^{-6}$&$4.0636\cdot10^{-15}$\\
    \hline
    $\Gamma_4$ & $0.0033$ &$0.0169$ \\
    \hline
    $\Gamma_5$ & $8.1731\cdot10^{-7}$&$5.9511\cdot 10^{-5}$\\
    \hline  
    $\Gamma_6$ & $1.2231\cdot10^{-8}$& \\
    \hline
    $\mathfrak{G}$ &  $9.125\cdot10^{-7}$&$0.0145$ \\
    \hline 
  \end{tabular}}
  \caption{Table of $p$-values for Student's $t$-test}
  \label{Table2}
\end{table}
\newpage

\begin{figure}[htbp!]
  \centering
  \begin{tikzpicture}[scale=0.55]
  \node at (12,10) {$\Gamma_1$};
    \begin{axis}[
      xlabel={Energy},
      ylabel={Frequency},
      ymin=0,
      ymax=20,
      xmin=0,
      xmax=3,
      width=0.8\linewidth,
     cycle list={
        {fill=blue!30, draw=blue!50!black},
        {fill=orange!30, draw=orange!50!black,fill opacity=0.5},
      },
    ]
      \addplot +[hist={data min=0, bins=20}]
        table [y index=0] {
          0.7197265625
          0.47265625
          0.7060546875
          0.951171875
          0.76171875
          0.693359375
          0.3583984375
          1.111328125
          1.162109375
          0.4345703125
          1.103515625
          0.8515625
          0.7138671875
          1.068359375
          0.6728515625
          0.673828125
          1.1064453125
          0.921875
          0.8251953125
          0.5615234375
          0.60546875
          0.744140625
          0.400390625
          0.4580078125
          0.8466796875
          0.4697265625
          0.802734375
          0.5439453125
          0.451171875
          0.5068359375
          1.0009765625
          0.6708984375
          0.474609375
          0.734375
          0.58203125
          1.1025390625
          0.97265625
          0.4443359375
          0.970703125
          0.4609375
          0.9208984375
          0.4873046875
          1.0146484375
          0.8544921875
          0.705078125
          0.5205078125
          0.6640625
          1.0341796875
          0.521484375
          0.462890625
        };
       
      \addplot +[hist={data min=0, bins=20}]
        table [y index=0] {
         0.4296875
         0.337890625
         0.453125
         0.4287109375
         0.587890625
         0.7451171875
         1.05078125
         0.4287109375
         0.2177734375
         0.1923828125
         0.94921875
         0.5126953125
         0.423828125
         0.3095703125
         0.6376953125
         0.47265625
         0.34375
         0.86328125
         0.5947265625
         0.5712890625
         0.458984375
         0.5810546875
         0.6552734375
         0.49609375
         0.6533203125
         0.6552734375
         0.4306640625
         0.439453125
         2.6904296875
         0.306640625
         0.4619140625
         0.3212890625
         0.71484375
         0.33984375
         0.564453125
         0.5341796875
         1.271484375
         0.890625
         0.4296875
         0.421875
         0.3603515625
         0.4365234375
         0.587890625
         0.53125
         0.5517578125
         0.3251953125
         0.5400390625
         0.4541015625
         0.4560546875
         0.34765625
        };
    \end{axis}
  \end{tikzpicture}\hspace{0.1in}
  \begin{tikzpicture}[scale=0.55]
   \node at (12,10) {$\Gamma_2$};
    \begin{axis}[
      xlabel={Energy},
      ylabel={Frequency},
      ymin=0,
      ymax=18,
      xmin=0,
      xmax=1.8,
      width=0.8\linewidth,
     cycle list={
        {fill=blue!30, draw=blue!50!black},
        {fill=orange!30, draw=orange!50!black,fill opacity=0.5},
      },
    ]
      \addplot +[hist={data min=0, bins=15}]
        table [y index=0] {
0.8935546875
0.7373046875
1.03125
0.9853515625
0.59375
0.8427734375
1.3671875
1.0283203125
0.7568359375
1.2744140625
0.8642578125
1.0576171875
1.3232421875
0.9296875
1.1103515625
1.345703125
0.98046875
0.9306640625
0.7158203125
1.3310546875
1.228515625
1.099609375
0.6767578125
0.8671875
1.3671875
0.7724609375
0.5341796875
0.7314453125
0.931640625
1.3564453125
1.2001953125
1.203125
1.1240234375
1.0009765625
1.2568359375
1.296875
0.91015625
0.798828125
1.2861328125
0.9072265625
0.73828125
0.7431640625
0.943359375
0.953125
0.6953125
0.931640625
0.7392578125
0.4814453125
0.6435546875
0.8115234375
0.865234375
0.7880859375
1.1826171875
0.5283203125
1.25390625
1.349609375

        };
       
      \addplot +[hist={data min=0, bins=15}]
        table [y index=0] {
1.056640625
0.947265625
0.724609375
0.69140625
0.7978515625
0.8427734375
0.728515625
1.103515625
0.8837890625
0.56640625
0.76171875
0.8798828125
0.7705078125
0.85546875
1.0478515625
0.6630859375
0.419921875
0.5029296875
0.96484375
0.6982421875
0.666015625
0.9130859375
0.4892578125
0.4345703125
0.634765625
0.5908203125
0.6484375
0.548828125
0.85546875
0.845703125
0.8125
0.896484375
1.0224609375
0.8525390625
0.5224609375
0.810546875
0.6953125
0.9267578125
0.50390625
0.443359375
0.873046875
0.76953125
0.646484375
0.6591796875
0.4189453125
0.728515625
0.94921875
0.67578125
1.140625
0.74609375
0.888671875
0.369140625
0.8662109375
0.5361328125
0.625
0.7314453125

};
    \end{axis}
  \end{tikzpicture}

  \begin{tikzpicture}[scale=0.55]
   \node at (12,10) {$\Gamma_3$};
     \begin{axis}[
      xlabel={Energy},
      ylabel={Frequency},
      ymin=0,
      ymax=24,
      xmin=0,
      xmax=3,
      width=0.8\linewidth,
     cycle list={
        {fill=blue!30, draw=blue!50!black},
        {fill=orange!30, draw=orange!50!black,fill opacity=0.5},
      },
    ]
      \addplot +[hist={data min=0, bins=15}]
        table [y index=0] {
0.7587890625
1.55078125
1.1533203125
1.376953125
1.3876953125
0.6533203125
1.0986328125
0.8359375
1.302734375
1.9658203125
1.2373046875
1.2978515625
1.6015625
1.24609375
1.1748046875
1.859375
1.287109375
1.2158203125
1.4091796875
1.583984375
1.2568359375
1.0625
1.1533203125
1.2177734375
1.2255859375
0.8515625
0.9443359375
0.8984375
2.0322265625
1.107421875
1.2724609375
1.0517578125
1.2353515625
1.3232421875
0.7734375
2.04296875
0.7099609375
1.4326171875
1.2724609375
1.1435546875
1.220703125
0.673828125
1.0654296875
0.9013671875
1.517578125
1.00390625
0.77734375
1.5888671875
1.376953125
1.0849609375
1.4072265625
1.4404296875
1.48828125
1.58984375
1.359375
1.4697265625
        };
       
      \addplot +[hist={data min=0, bins=15}]
        table [y index=0] {
0.99609375
1.8056640625
1.8330078125
1.83984375
0.6435546875
1.412109375
0.908203125
1.021484375
1.857421875
1.0087890625
1.0
0.73828125
0.4228515625
0.49609375
0.7626953125
1.2744140625
0.626953125
1.0380859375
0.732421875
0.6943359375
0.8876953125
1.0048828125
0.8916015625
0.84375
1.080078125
1.1171875
0.783203125
0.75
1.119140625
0.984375
1.0966796875
1.0576171875
0.96484375
1.501953125
0.474609375
0.83203125
0.9482421875
0.841796875
1.1787109375
1.0810546875
0.8466796875
0.6689453125
1.12109375
0.767578125
0.625
0.880859375
0.8837890625
0.5546875
1.01171875
0.794921875
0.5693359375
0.8076171875
0.369140625
0.7451171875
0.5517578125
0.576171875
};
    \end{axis}
  \end{tikzpicture}\hspace{0.1in}
\begin{tikzpicture}[scale=0.55]
 \node at (12,10) {$\Gamma_4$};
    \begin{axis}[
      xlabel={Energy},
      ylabel={Frequency},
      ymin=0,
      ymax=20,
      xmin=0,
      xmax=4,
      width=0.8\linewidth,
     cycle list={
        {fill=blue!30, draw=blue!50!black},
        {fill=orange!30, draw=orange!50!black,fill opacity=0.5},
      },
    ]
      \addplot +[hist={data min=0, bins=15}]
        table [y index=0] {
1.9794921875
1.76171875
1.90234375
2.0283203125
1.8857421875
1.61328125
1.4912109375
1.8388671875
1.67578125
1.6982421875
1.2509765625
1.087890625
1.53515625
0.9619140625
1.6357421875
0.9931640625
1.2158203125
1.318359375
1.4931640625
1.01171875
1.4384765625
1.6015625
1.1884765625
0.73828125
1.9501953125
1.3408203125
0.9677734375
1.86328125
1.93359375
1.52734375
2.2734375
1.4296875
1.052734375
1.82421875
1.5869140625
1.576171875
1.5947265625
0.9580078125
2.140625
1.41015625
0.9755859375
1.423828125
1.53515625
1.427734375
1.6650390625
1.1162109375
1.30078125
1.7001953125
1.05078125
1.2109375
1.07421875
1.591796875
1.0
2.087890625
1.6591796875
1.60546875
        };
       
      \addplot +[hist={data min=0, bins=15}]
        table [y index=0] {
1.337890625
4.0205078125
1.150390625
1.103515625
1.1376953125
1.3916015625
0.78125
1.111328125
1.150390625
1.1142578125
1.1484375
0.8193359375
1.2724609375
0.720703125
1.2490234375
1.2998046875
0.72265625
0.8916015625
1.3759765625
1.763671875
0.7705078125
1.509765625
0.697265625
0.771484375
1.623046875
1.0498046875
0.7705078125
1.03515625
1.5703125
0.8212890625
0.7666015625
1.5078125
1.5185546875
0.7626953125
1.1513671875
1.0986328125
1.06640625
0.5556640625
0.9345703125
1.330078125
1.3232421875
4.037109375
1.4052734375
1.0908203125
1.3017578125
1.5078125
1.091796875
1.5419921875
0.8427734375
0.51171875
0.7998046875
1.43359375
0.775390625
0.755859375
0.73828125
0.939453125
};
    \end{axis}
  \end{tikzpicture}
    \begin{tikzpicture}[scale=0.55]
     \node at (12,10) {$\Gamma_5$};
    \begin{axis}[
      xlabel={Energy},
      ylabel={Frequency},
      ymin=0,
      ymax=20,
      xmin=0,
      xmax=4,
      width=0.8\linewidth,
     cycle list={
        {fill=blue!30, draw=blue!50!black},
        {fill=orange!30, draw=orange!50!black,fill opacity=0.5},
      },
    ]
      \addplot +[hist={data min=0, bins=15}]
        table [y index=0] {
          0.8876953125
          1.2763671875
          1.3173828125
          0.8857421875
          0.67578125
          1.2724609375
          1.3515625
          1.04296875
          1.8623046875
          2.021484375
          1.744140625
          2.0673828125
          1.5673828125
          1.3984375
          1.216796875
          1.748046875
          0.90625
          1.361328125
          1.177734375
          1.7314453125
          1.146484375
          0.9697265625
          1.4765625
          0.875
          1.423828125
          1.251953125
          0.6162109375
          3.7978515625
          0.90234375
          1.173828125
          0.7646484375
          0.787109375
          2.12109375
          1.5830078125
          1.3203125
          0.8544921875
          0.8984375
          1.5478515625
          0.771484375
          1.2041015625
          1.7158203125
          1.455078125
          1.5888671875
          1.1015625
          1.6943359375
          1.44140625
          1.2958984375
          0.9619140625
          1.3642578125
          1.7265625
        };
       
      \addplot +[hist={data min=0, bins=15}]
        table [y index=0] {
          0.5556640625
          1.0712890625
          0.814453125
          0.9384765625
          0.5458984375
          1.19921875
          0.697265625
          1.0419921875
          1.1787109375
          1.076171875
          0.7392578125
          0.96484375
          0.8251953125
          0.9091796875
          0.3515625
          1.1796875
          1.2470703125
          1.1689453125
          1.0224609375
          0.5908203125
          0.962890625
          0.98828125
          0.5693359375
          0.88671875
          1.0869140625
          0.5478515625
          1.2666015625
          0.755859375
          0.7587890625
          0.87109375
          1.005859375
          0.7548828125
          0.7685546875
          1.0107421875
          1.0537109375
          1.90625
          0.7744140625
          1.0517578125
          0.73828125
          0.8505859375
          0.9853515625
          0.6572265625
          0.7197265625
          1.0859375
          1.1904296875
          0.451171875
          1.01171875
          1.232421875 
          0.9833984375 
          0.6982421875
        };
    \end{axis} 
  \end{tikzpicture}\hspace{0.1in}
  \begin{tikzpicture}[scale=0.55]
   \node at (12,10) {$\Gamma_6$};
    \begin{axis}[
      xlabel={Energy},
      ylabel={Frequency},
      ymin=0,
      ymax=9,
      xmin=0,
      xmax=1.5,
      width=0.8\linewidth,
     cycle list={
        {fill=blue!30, draw=blue!50!black},
        {fill=orange!30, draw=orange!50!black,fill opacity=0.5},
      },
    ]
    
      \addplot +[hist={data min=0,bins=15}]
        table [y index=0] {
            1.2734375
            0.8828125
            0.9345703125
            0.7607421875
            0.8056640625
            0.8671875
            1.029296875
            0.9658203125
            0.9189453125
            0.677734375
            0.923828125
            1.4697265625
            0.943359375
            0.50390625
            0.7412109375
            0.80078125
            1.2392578125
            0.619140625
            0.8486328125
            0.865234375
            0.5771484375
            0.908203125
            0.7236328125
            0.818359375
            1.2490234375
            0.8076171875
            0.7763671875
            0.501953125
        };
       
      \addplot +[hist={data min=0,bins=15}]
        table [y index=0] {
            0.740234375
            0.5732421875
            0.4287109375
            0.8701171875
            0.484375
            0.8056640625
            0.4775390625
            0.5478515625
            0.548828125
            0.431640625
            0.62109375
            0.5361328125
            0.359375
            0.3427734375
            0.634765625
            0.361328125
            0.5654296875
            0.4560546875
            0.509765625
            0.771484375
            0.71484375
            0.4404296875
            0.4228515625
            0.5048828125
            0.556640625
            0.3310546875
            0.169921875
            0.4287109375
        };
    \end{axis}
  \end{tikzpicture}

  \begin{tikzpicture}[scale=0.55]
   \node at (12,10) {$\mathfrak{G}$};
    \begin{axis}[
      xlabel={Energy},
      ylabel={Frequency},
      ymin=0,
      ymax=14,
      xmin=4,
      xmax=15,
      width=0.8\linewidth,
     cycle list={
        {fill=blue!30, draw=blue!50!black},
        {fill=orange!30, draw=orange!50!black,fill opacity=0.5},
      },
    ]
      \addplot +[hist={data min=0, bins=20}]
        table [y index=0] {
       6.244140625
       11.4814453125
       7.03125
       13.751953125
       10.3359375
       11.341796875
       8.283203125
       10.568359375
       8.0009765625
       8.5439453125
       9.2236328125
       11.728515625
       9.7509765625
       8.8525390625
       6.7548828125
       11.6884765625
       10.36328125
       8.052734375
       7.7373046875
       9.3095703125
       11.6220703125
       11.138671875
       9.30078125
       9.0810546875
       8.375
       9.818359375
       10.845703125
       10.0556640625
       11.369140625
       12.1640625
       11.9833984375
       11.3623046875
       10.6640625
       10.3486328125
       9.7333984375
       9.09765625
       9.6328125
       11.68359375
       10.5478515625
       8.7705078125
       10.8388671875
       8.5537109375
       10.9375
       10.8603515625
       9.2939453125
       13.6357421875
       10.5087890625
       11.9189453125
       10.1201171875
       10.759765625
        };
       
      \addplot +[hist={data min=0, bins=20}]
        table [y index=0] {
         11.1982421875
         10.380859375
         10.140625
         8.1171875
         6.310546875
         9.2841796875
         11.1376953125
         5.552734375
         8.103515625
         10.216796875
         6.7373046875
         5.8720703125
         7.8486328125
         7.1875
         9.87109375
         4.4765625
         9.70703125
         5.384765625
         4.5400390625
         6.044921875
         9.1015625
         6.666015625
         9.892578125
         10.5
         7.5126953125
         5.4208984375
         9.3388671875
         5.251953125
         4.9921875
         4.9287109375
         5.7861328125
         5.6337890625
         9.662109375
         7.44140625
         7.94140625
         5.05859375
         9.96484375
         8.8173828125
         9.5810546875
         9.0849609375
         11.908203125
         12.74609375
         8.8642578125
         9.662109375
         6.443359375
         10.765625
         5.896484375
         6.01953125
         10.4931640625
         9.1298828125
        };
    \end{axis}
  \end{tikzpicture}

  \caption{Histograms illustrating the frequency distributions of $\mathcal{E}p$-values for algorithms utilizing mixers $H_M$ and $H_M$}
  \label{Histograms1}
\end{figure}

\begin{figure}[htbp!]
  \centering
 \begin{tikzpicture}[scale=0.55]
 \node at (12,10) {$\Gamma_1$};
    \begin{axis}[
      xlabel={Energy},
      ylabel={Frequency},
      ymin=0,
      ymax=20,
      xmin=0,
      xmax=2.4,
      width=0.8\linewidth,
     cycle list={
        {fill=blue!30, draw=blue!50!black},
        {fill=orange!30, draw=orange!50!black,fill opacity=0.5},
      },
    ]
      \addplot +[hist={data min=0, bins=20}]
        table [y index=0] {
          0.7197265625
          0.47265625
          0.7060546875
          0.951171875
          0.76171875
          0.693359375
          0.3583984375
          1.111328125
          1.162109375
          0.4345703125
          1.103515625
          0.8515625
          0.7138671875
          1.068359375
          0.6728515625
          0.673828125
          1.1064453125
          0.921875
          0.8251953125
          0.5615234375
          0.60546875
          0.744140625
          0.400390625
          0.4580078125
          0.8466796875
          0.4697265625
          0.802734375
          0.5439453125
          0.451171875
          0.5068359375
          1.0009765625
          0.6708984375
          0.474609375
          0.734375
          0.58203125
          1.1025390625
          0.97265625
          0.4443359375
          0.970703125
          0.4609375
          0.9208984375
          0.4873046875
          1.0146484375
          0.8544921875
          0.705078125
          0.5205078125
          0.6640625
          1.0341796875
          0.521484375
          0.462890625
        };
       
      \addplot +[hist={data min=0, bins=20}]
        table [y index=0] {
         0.7509765625
         0.466796875
         0.6162109375
         0.4423828125
         0.9833984375
         0.66015625
         0.45703125
         2.23046875
         0.4990234375
         0.162109375
         0.6943359375
         0.44921875
         0.3828125
         0.5478515625
         0.509765625
         0.572265625
         1.1337890625
         0.6220703125
         0.568359375
         0.3388671875
         0.4521484375
         0.572265625
         0.482421875
         0.7216796875
         0.4609375
         0.9482421875
         0.53125
         0.64453125
         0.419921875
         0.5185546875
         0.3173828125
         0.5400390625
         0.5546875
         0.603515625
         0.552734375
         0.4140625
         0.396484375
         0.333984375
         0.56640625
         0.451171875
         0.6630859375
         0.3564453125
         0.259765625
         1.30078125 
         0.318359375 
         0.2587890625
         0.427734375 
         0.357421875
         0.4814453125
         0.634765625
        };
    \end{axis}
  \end{tikzpicture}
 \begin{tikzpicture}[scale=0.55]
 \node at (12,10) {$\Gamma_2$};
    \begin{axis}[
      xlabel={Energy},
      ylabel={Frequency},
      ymin=0,
      ymax=25,
      xmin=0,
      xmax=3.5,
      width=0.8\linewidth,
     cycle list={
        {fill=blue!30, draw=blue!50!black},
        {fill=orange!30, draw=orange!50!black,fill opacity=0.5},
      },
    ]
      \addplot +[hist={data min=0, bins=15}]
        table [y index=0] {
0.8935546875
0.7373046875
1.03125
0.9853515625
0.59375
0.8427734375
1.3671875
1.0283203125
0.7568359375
1.2744140625
0.8642578125
1.0576171875
1.3232421875
0.9296875
1.1103515625
1.345703125
0.98046875
0.9306640625
0.7158203125
1.3310546875
1.228515625
1.099609375
0.6767578125
0.8671875
1.3671875
0.7724609375
0.5341796875
0.7314453125
0.931640625
1.3564453125
1.2001953125
1.203125
1.1240234375
1.0009765625
1.2568359375
1.296875
0.91015625
0.798828125
1.2861328125
0.9072265625
0.73828125
0.7431640625
0.943359375
0.953125
0.6953125
0.931640625
0.7392578125
0.4814453125
0.6435546875
0.8115234375
0.865234375
0.7880859375
1.1826171875
0.5283203125
1.25390625
1.349609375
        };
       
      \addplot +[hist={data min=0, bins=15}]
        table [y index=0] {
0.57421875
1.359375
0.9267578125
0.5166015625
2.634765625
1.21484375
0.5166015625
0.705078125
0.7509765625
0.8310546875
0.5078125
0.396484375
0.720703125
0.681640625
0.80859375
0.947265625
0.541015625
0.630859375
3.3291015625
0.998046875
0.6767578125
0.939453125
0.708984375
0.498046875
0.8095703125
0.66015625
0.6796875
0.546875
1.3115234375
0.720703125
0.5595703125
0.7919921875
0.86328125
3.357421875
0.705078125
0.7841796875
0.69140625
1.080078125
0.7685546875 
0.908203125
0.4140625
0.794921875
0.6171875
0.73828125
0.5966796875
0.9150390625 
0.5263671875
0.603515625
0.6123046875
0.5517578125
0.580078125
1.0693359375
0.75390625
0.4228515625
1.1865234375
0.6171875

};
    \end{axis}
  \end{tikzpicture}
 \begin{tikzpicture}[scale=0.55]
   \node at (12,10) {$\Gamma_3$};
     \begin{axis}[
      xlabel={Energy},
      ylabel={Frequency},
      ymin=0,
      ymax=24,
      xmin=0,
      xmax=3,
      width=0.8\linewidth,
     cycle list={
        {fill=blue!30, draw=blue!50!black},
        {fill=orange!30, draw=orange!50!black,fill opacity=0.5},
      },
    ]
      \addplot +[hist={data min=0, bins=15}]
        table [y index=0] {
0.7587890625
1.55078125
1.1533203125
1.376953125
1.3876953125
0.6533203125
1.0986328125
0.8359375
1.302734375
1.9658203125
1.2373046875
1.2978515625
1.6015625
1.24609375
1.1748046875
1.859375
1.287109375
1.2158203125
1.4091796875
1.583984375
1.2568359375
1.0625
1.1533203125
1.2177734375
1.2255859375
0.8515625
0.9443359375
0.8984375
2.0322265625
1.107421875
1.2724609375
1.0517578125
1.2353515625
1.3232421875
0.7734375
2.04296875
0.7099609375
1.4326171875
1.2724609375
1.1435546875
1.220703125
0.673828125
1.0654296875
0.9013671875
1.517578125
1.00390625
0.77734375
1.5888671875
1.376953125
1.0849609375
1.4072265625
1.4404296875
1.48828125
1.58984375
1.359375
1.4697265625
        };
       
      \addplot +[hist={data min=0, bins=15}]
        table [y index=0] {
1.560546875
0.9970703125
0.8369140625
0.6494140625
0.7734375
1.125
0.5458984375
0.5458984375
0.6767578125
0.36328125
0.6123046875
0.626953125
0.8505859375
0.9521484375
0.8095703125
0.5498046875
0.701171875
0.62109375
0.3193359375
0.6328125
0.869140625
0.6708984375
0.48046875
0.8056640625
0.8564453125
0.9853515625
0.52734375
0.5791015625
1.544921875
0.5908203125
0.685546875
0.5751953125
0.5126953125
1.0576171875
1.541015625
0.9013671875
0.73828125
0.75390625
0.8974609375
0.958984375
0.513671875
0.4990234375
0.5185546875
0.85546875
0.552734375
0.5048828125
0.931640625
0.4443359375
0.259765625
0.5
1.0419921875
0.5078125
0.4580078125
0.7392578125
0.67578125
0.783203125
};
\end{axis}
\end{tikzpicture}
   \begin{tikzpicture}[scale=0.55]
   \node at (12,10) {$\Gamma_4$};
    \begin{axis}[
      xlabel={Energy},
      ylabel={Frequency},
      ymin=0,
      ymax=20,
      xmin=0,
      xmax=4,
      width=0.8\linewidth,
     cycle list={
        {fill=blue!30, draw=blue!50!black},
        {fill=orange!30, draw=orange!50!black,fill opacity=0.5},
      },
    ]
      \addplot +[hist={data min=0, bins=15}]
        table [y index=0] {
1.9794921875
1.76171875
1.90234375
2.0283203125
1.8857421875
1.61328125
1.4912109375
1.8388671875
1.67578125
1.6982421875
1.2509765625
1.087890625
1.53515625
0.9619140625
1.6357421875
0.9931640625
1.2158203125
1.318359375
1.4931640625
1.01171875
1.4384765625
1.6015625
1.1884765625
0.73828125
1.9501953125
1.3408203125
0.9677734375
1.86328125
1.93359375
1.52734375
2.2734375
1.4296875
1.052734375
1.82421875
1.5869140625
1.576171875
1.5947265625
0.9580078125
2.140625
1.41015625
0.9755859375
1.423828125
1.53515625
1.427734375
1.6650390625
1.1162109375
1.30078125
1.7001953125
1.05078125
1.2109375
1.07421875
1.591796875
1.0
2.087890625
1.6591796875
1.60546875
        };
       
      \addplot +[hist={data min=0, bins=15}]
        table [y index=0] {
0.5810546875
1.6318359375
1.1279296875
1.3994140625
1.15234375
0.984375
1.587890625
0.5703125
1.2109375
1.6142578125
1.646484375
0.7333984375
0.8291015625
0.8759765625
1.892578125
1.4560546875
0.95703125
1.1044921875
1.314453125
1.44140625
0.8701171875
0.9189453125
2.033203125
0.7548828125
0.73828125
0.8232421875
0.7509765625
1.126953125
1.04296875
1.16796875
0.478515625
1.5283203125
1.3818359375
1.44921875
0.6689453125
0.458984375
1.1455078125
4.06640625
1.361328125
0.8427734375
1.205078125
0.5595703125
1.33984375
1.12109375
1.4677734375
4.0927734375
1.125
1.2421875
1.1181640625
0.439453125
1.2666015625
1.5205078125
1.4794921875
1.6552734375
1.26953125
0.9033203125

};
    \end{axis}
  \end{tikzpicture}

    \begin{tikzpicture}[scale=0.55]
    \node at (12,10) {$\Gamma_5$};
    \begin{axis}[
      xlabel={Energy},
      ylabel={Frequency},
      ymin=0,
      ymax=20,
      xmin=0,
      xmax=4,
      width=0.8\linewidth,
     cycle list={
        {fill=blue!30, draw=blue!50!black},
        {fill=orange!30, draw=orange!50!black,fill opacity=0.5},
      },
    ]
      \addplot +[hist={data min=0, bins=15}]
        table [y index=0] {
           0.8876953125
          1.2763671875
          1.3173828125
          0.8857421875
          0.67578125
          1.2724609375
          1.3515625
          1.04296875
          1.8623046875
          2.021484375
          1.744140625
          2.0673828125
          1.5673828125
          1.3984375
          1.216796875
          1.748046875
          0.90625
          1.361328125
          1.177734375
          1.7314453125
          1.146484375
          0.9697265625
          1.4765625
          0.875
          1.423828125
          1.251953125
          0.6162109375
          3.7978515625
          0.90234375
          1.173828125
          0.7646484375
          0.787109375
          2.12109375
          1.5830078125
          1.3203125
          0.8544921875
          0.8984375
          1.5478515625
          0.771484375
          1.2041015625
          1.7158203125
          1.455078125
          1.5888671875
          1.1015625
          1.6943359375
          1.44140625
          1.2958984375
          0.9619140625
          1.3642578125
          1.7265625
        };
       \addplot +[hist={data min=0, bins=15}]
        table [y index=0] {
          0.7705078125
          1.015625
          0.6357421875
          1.1220703125
          1.1748046875
          1.021484375
          0.7626953125
          0.84375
          1.0478515625
          0.7451171875
          0.2939453125
          0.7197265625
          3.564453125
          0.6787109375
          0.943359375
          1.03125
          1.0517578125
          1.2998046875
          0.5810546875
          0.419921875 
          1.072265625
          0.98828125 
          1.052734375
          1.1630859375
          0.775390625
          1.158203125
          1.0888671875 
          0.92578125
          1.00390625
          0.7197265625
          1.2841796875
          1.115234375
          1.0029296875
          0.4296875 
          0.9267578125
          1.1337890625
          1.0966796875
          1.263671875
          0.7626953125
          1.0029296875
          0.6748046875 
          0.8154296875
          1.0615234375
          0.7060546875
          0.4638671875
          0.857421875 
          0.46484375
          0.6328125 
          0.931640625 
          0.7626953125
        };
    \end{axis}
  \end{tikzpicture}
\begin{tikzpicture}[scale=0.55]
  \node at (12,10) {$\mathfrak{G}$};
    \begin{axis}[
      xlabel={Energy},
      ylabel={Frequency},
      ymin=0,
      ymax=14,
      xmin=4,
      xmax=15,
      width=0.8\linewidth,
     cycle list={
        {fill=blue!30, draw=blue!50!black},
        {fill=orange!30, draw=orange!50!black,fill opacity=0.5},
      },
    ]
      \addplot +[hist={data min=0, bins=20}]
        table [y index=0] {
       6.244140625
       11.4814453125
       7.03125
       13.751953125
       10.3359375
       11.341796875
       8.283203125
       10.568359375
       8.0009765625
       8.5439453125
       9.2236328125
       11.728515625
       9.7509765625
       8.8525390625
       6.7548828125
       11.6884765625
       10.36328125
       8.052734375
       7.7373046875
       9.3095703125
       11.6220703125
       11.138671875
       9.30078125
       9.0810546875
       8.375
       9.818359375
       10.845703125
       10.0556640625
       11.369140625
       12.1640625
       11.9833984375
       11.3623046875
       10.6640625
       10.3486328125
       9.7333984375
       9.09765625
       9.6328125
       11.68359375
       10.5478515625
       8.7705078125
       10.8388671875
       8.5537109375
       10.9375
       10.8603515625
       9.2939453125
       13.6357421875
       10.5087890625
       11.9189453125
       10.1201171875
       10.759765625
        };
       
      \addplot +[hist={data min=0, bins=20}]
        table [y index=0] {
         9.732421875
         7.076171875
         6.4013671875
         5.33984375
         10.7353515625
         8.7021484375
         8.9619140625
         13.6474609375
         6.12109375
         8.5380859375
         7.9765625
         7.875
         13.1259765625
         5.63671875
         11.36328125
         7.5361328125
         8.2001953125
         9.3388671875
         9.4052734375
         6.982421875
         10.7568359375
         8.822265625
         5.630859375
         7.5478515625
         8.0595703125
         6.431640625
         10.6962890625
         9.6943359375
         8.9541015625
         10.05078125
         8.5751953125
         12.6513671875
         13.6552734375
         9.2294921875
         9.228515625
         4.9443359375
         7.71484375
         7.28125
         5.3447265625
         5.65234375
         13.5234375
         13.6474609375
         12.0439453125
         7.6533203125
         10.59765625
         9.6298828125
         9.1064453125
         12.6533203125
         12.947265625
         6.5439453125
        };
    \end{axis}
  \end{tikzpicture}
\caption{Histograms illustrating the frequency distributions of $\mathcal{E}p$-values for algorithms utilizing mixers $H_M$ and $H_\chi$}
  \label{Histograms2}
\end{figure}
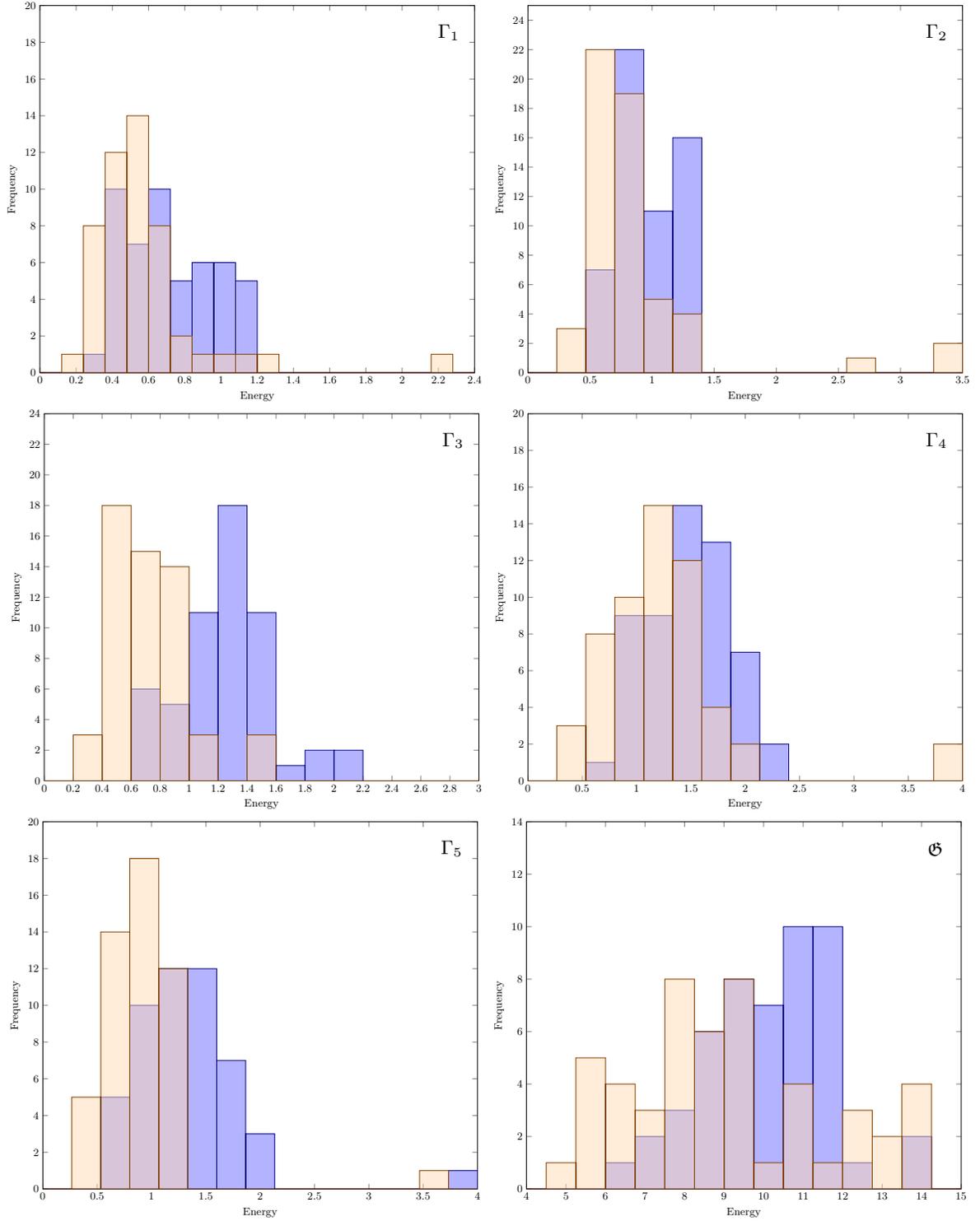
\newpage

\section{Mixer Hamiltonians for Warm-Start QAOA}

The standard QAOA  typically begins in the uniform superposition of all classical bit strings. Its primary objective is to enhance the objective function's value beyond the expected value in this initial state. A natural extension involves running a classical algorithm to generate a promising string (i.e., a good solution for certain practical goals, e.g., optimizing time/quality trade-off), then initializing the QAOA in the corresponding computational basis state to seek further improvement. This approach, known as warm-start QAOA, has been explored in various studies.

In a recent paper \cite{CFGRT}, extensive numerical experiments involving both small and large instances at varying depths revealed a notable observation. Specifically, when the QAOA commences from a single warm-start string (or superposition of strings with equal energy level), it exhibits negligible progress. The authors also highlight that these findings hold true even when the QAOA initializes with a single classical string, and the unitary operators constituting the QAOA do not explicitly rely on the initial string.

We aim to further explore this topic by offering additional insights into the limitations of warm-start QAOA. We start by observing that any mixer Hamiltonian with a nontrivial spectral gap possesses a one-dimensional eigenspace corresponding to the smallest eigenvalue \( \lambda \). Let \( |s\rangle \) denote a state spanning this subspace. It follows that \( |s\rangle \) cannot be an eigenvector for the problem Hamiltonian \( H_P \). If it were, both Hamiltonians would merely scale \( |s\rangle \), and executing the corresponding QAOA starting with the state \( |s\rangle \) would result in the identical state (up to a phase). Subsequently, measuring in the standard basis would yield a standard state with energy \( \lambda \).

In addition, as emphasized in the second assertion of Theorem \ref{PF}, every irreducible matrix with nonnegative values ensures that the vector corresponding to the highest eigenvalue has coordinates that are all nonzero (positive) in the standard basis.

This crucially implies that the ground state for a mixer Hamiltonian, satisfying the assumptions of the Perron-Frobenius theorem, \emph{must be} a superposition of all classical states with nonzero amplitudes. Consequently, the standard argument for guaranteeing the convergence of QAOA as $p\rightarrow\infty$ to an optimal classical solution is inapplicable unless the initial state is a superposition of all classical states with nonzero amplitudes.

\section{Appendix}

In this section, we introduce and detail the construction, matrix representation, and quantum circuit for the newly introduced mixer Hamiltonian $H_M$, which is employed throughout the paper. We demonstrate that $H_M$ satisfies the Perron-Frobenius theorem, ensuring the convergence of the corresponding Quantum Approximate Optimization Algorithm as the number of iterations, $p$, tends to infinity. Additionally, we revisit essential definitions and provide a concrete representation of the result concerning subgroups of $\mathcal{S}$ that commute with the actions of the operators $H_M$ and $B$ on $W$.

\subsection{New Mixers: $H_M$ and $H_\chi$}

The symmetric group $S_d$ discussed in previous sections is also known as $W(U_d)$, the Weyl subgroup of the unitary group acting collectively on all qudits. A mixer Hamiltonian $H_M$, which commutes with this group's action, can be constructed as follows.

Consider the sum of all transpositions, $\zeta = \sum\limits_{1 \leq i < j \leq n} (ij) \in \mathbb{C}[S_d]$, where $\mathbb{C}[S_d]$ denotes the group algebra of $S_d$. The group algebra is a vector space with a basis indexed by group elements, where multiplication is defined by the group operation of the underlying group. This element $\zeta$ commutes with all permutations and, therefore, resides in the center of the group algebra. The matrix representation of $\zeta$ in the standard basis of a vector space representing a single qudit is given by $\widehat{H}_{M_{ij}} = 1$ for $i \neq j$, and $\widehat{H}_{M_{ii}} = {d-1 \choose 2}$.  A notable practical observation is that, in the Hadamard basis, this matrix becomes diagonal:

\[
H^{\otimes \ell}\widehat{H}_M H^{\otimes \ell}=\text{diag}\left(\frac{d(d-1)}{2},\frac{(d-1)(d-2)}{2}-1,\ldots,\frac{(d-1)(d-2)}{2}-1\right),
\]

or, ignoring the addition of a scalar $\left(\frac{(d-1)(d-2)}{2}-1\right)\cdot \text{Id}$ operator:

\[
H^{\otimes \ell}\widehat{H}_M H^{\otimes \ell}=\text{diag}(d,0,\ldots,0),
\]

resulting in $e^{-\beta H^{\otimes \ell}\widehat{H}_M H^{\otimes \ell}}=\text{diag}(e^{-d\beta},1,\ldots,1)$. 

\begin{ex}
    The quantum circuits for $e^{-\beta H^{\otimes \ell}\widehat{H}_M H^{\otimes \ell}}$ with $d=4$ and $d=8$ are illustrated on Figure \ref{MixerCircuit}.

\begin{figure}[htbp!]
\begin{quantikz}
e_0\hspace{0.05in} & \qw &\gate{H} & \targ{}  &\ctrl{1} &  \targ{}   &\gate{H} & \qw   \\
e_1\hspace{0.05in} & \qw & \gate{H} & \targ{} 	 & \gate{P_{-4\beta}} & \targ{}  & \gate{H} & \qw\\
\vphantom{\rule{0pt}{0.23in}abc}\\
\end{quantikz}\hspace{0.1in} \begin{quantikz}
e_0\hspace{0.05in} & \qw &\gate{H} & \targ{}  &\ctrl{2} &  \targ{}   &\gate{H} & \qw   \\
e_1\hspace{0.05in} & \qw & \gate{H} & \targ{} 	 & \ctrl{1} & \targ{}  & \gate{H} & \qw\\
e_2\hspace{0.05in} & \qw & \gate{H} & \targ{} 	 & \gate{P_{-8\beta}} & \targ{}  & \gate{H} & \qw\\
\end{quantikz}
\caption{Quantum circuit for $e^{-\beta H^{\otimes \ell}\widehat{H}_M H^{\otimes \ell}}$ with $d=4$ and $d=8$}
\label{MixerCircuit}
\end{figure}
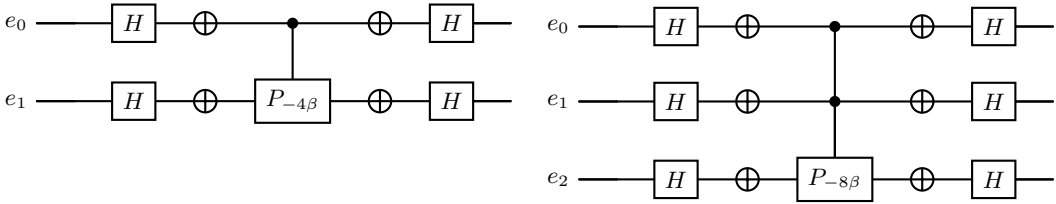
\end{ex}

We define the mixer Hamiltonian $H_M$ as the sum of individual terms $\widehat{H}^i_M$, where $\widehat{H}^i_M := Id\otimes\ldots\otimes Id\otimes\widehat{H}_M\otimes Id\ldots Id$ represents the action of $H_M$ on the $i^{th}$ copy of the vector space $V$ corresponding to the $i^{th}$ copy of the set $\mathbb{D}$. We would like to remind the reader  that $W = V \otimes \ldots \otimes V$ is the $n$-fold tensor product of such vector spaces.

It is straightforward to verify that the operator $H_M$ defined in this manner satisfies the assumptions of the Perron-Frobenius theorem (see \ref{PF}), thereby qualifying as a mixer Hamiltonian. Notably, its ground state $|\xi\rangle$ coincides with that of the classical mixer $B$.
 
This time, we start with the element $\eta = \sum_{1 \leq i < j \leq n} (-1)^{i+j}(ij) \in \mathbb{C}[S_d]$. We then examine the cyclic subgroup $\mathbb{Z}_d \subset S_d$, generated by the element $g = (23 \ldots n1)$, which cyclically permutes the elements from $1$ to $n$.

 \begin{lem}
 The element $\eta$ commutes with the group $\mathbb{Z}_d$.
 \end{lem}
 
\begin{proof}
To demonstrate this, we calculate $g\eta g^{-1} = \sum\limits_{1 \leq i < j \leq n} (-1)^{i+j}(g(i)g(j)) = \sum\limits_{1 \leq i < j \leq n} (-1)^{i+j}((i+1) \mod d)((j+1) \mod d) = \sum\limits_{1 \leq i < j \leq n} (-1)^{i+j}(ij) = \eta$, where $i+j+2 \equiv i+j \pmod{2}$. Hence, $g\eta g^{-1} = \eta$ implies $g\eta = \eta g$, indicating that $\eta$ and $g$ commute.
\end{proof}
  The matrix representation of $\eta$ in the standard basis of a vector space representing a single qudit is given by $\widehat{H}_{\chi_{ij}} = (-1)^{i+j}$ for $i \neq j$, and $\widehat{H}_{\chi_{ii}} = {d-1 \choose 2}$.  Furthermore, in the Hadamard basis, this matrix becomes diagonal:

\[
H^{\otimes \ell}\widehat{H}_\chi H^{\otimes \ell}=\text{diag}\left(\frac{(d-1)(d-2)}{2}-1,\ldots,\frac{(d-1)(d-2)}{2}-1,\frac{d(d-1)}{2},\frac{(d-1)(d-2)}{2}-1,\ldots,\frac{(d-1)(d-2)}{2}-1\right),
\]

or, ignoring the addition of a scalar $\left(\frac{(d-1)(d-2)}{2}-1\right)\cdot \text{Id}$ operator:

\[
H^{\otimes \ell}\widehat{H}_\chi H^{\otimes \ell}=\text{diag}(0,\ldots,0,d,0,\ldots,0), 
\]

acting with multiplication by $d$ on the one-dimensional vector space spanned by the vector $|-\underbrace{+\ldots++}_{\ell-1}\rangle$, 

resulting in $e^{-\beta H^{\otimes \ell}\widehat{H}_\chi H^{\otimes \ell}}=\text{diag}(1,\ldots,1,e^{-d\beta},1,\ldots,1)$. 

We define the mixer Hamiltonian $H_\chi$ as the sum $\widehat{H}_\chi\otimes Id\ldots\otimes Id+\sum\limits_{i=2}^n \widehat{H}^i_M$.

The operator $H_\chi$, defined in this manner, does not meet the assumptions of the Perron-Frobenius theorem (see \ref{PF}). Nevertheless, it is straightforward to verify that it possesses a one-dimensional eigenspace with the minimal eigenvalue, and therefore, a nonzero spectral gap. This eigenspace is spanned by the state $|\psi\rangle:=|-\underbrace{+\ldots++}_{n\ell-1}\rangle$.

Let $\zeta := e^{\frac{2\pi i}{d}}$ be the primitive $d$th root of unity. Under the action of the cyclic group $\mathbb{Z}_d$, the Hilbert space $W$ decomposes into a direct sum of vector spaces: $W = \bigoplus_{j=0}^{d-1} W_j$, where each $W_j$ has dimension $d^{n-1}$. The action of $\mathbb{Z}_d$ on $W_j$ is given by the equation $g \cdot w_j = \zeta^j w_j$, for all $w_j \in W_j$.

It is worth noting that $\zeta^{d/2} = e^{\pi i} = -1$, and the vector $|\psi\rangle$ resides in $W_{d/2}$.

\begin{rmk}
Since the actions of both operators \(H_P\) and \(H_\chi\) on \(W\) commute with that of the group \(\mathbb{Z}_d\), it follows that the operator \(\mathfrak{Q}_p\) preserves each subspace \(W_j\), meaning \(\mathfrak{Q}_p(W_j)\subseteq W_j\). Specifically, this implies \(\mathfrak{Q}_p(|\xi\rangle)\subseteq W_0\) and \(\mathfrak{Q}_p(|\psi\rangle)\subseteq W_j\).

    \label{subspacermk}
\end{rmk}

\subsection{Group actions and mixers}

Recall that the symmetric group $\mathcal{S}$ acts on  the set of all states $\mathbb{D}^n$ by permutations. This action can be uniquely extended to a linear action on the state vector space  $W$. Said differently, there is a homomorphism $\varphi: \mathcal{S} \rightarrow GL(W)$.

We elucidate key properties concerning the interaction of $S_d$, $K_\ell$, and $S_\ell$ (see Section $3$ for the definitions of these groups) with the objective function $F$ and the Hamiltonians $H_P$, $B$, and $H_M$.

Suppose $A: V \rightarrow V$ is a linear operator. We denote by $Z_{S_{d}}(A)$ the subgroup of elements in $S_d$ whose action on $V$ commutes with that of $A$.

\begin{prop} Suppose the objective function $F:\mathbb{D}^n\rightarrow \mathbb{R}$ is invariant with respect to the action of symmetric group $S_d$.

\begin{enumerate}
\item[$(a)$] $Z_{S_{d}}(H_P)=S_d$
\item[$(b)$] $Z_{S_{d}}(H_M)=S_d$
\item[$(c)$] $Z_{S_{d}}(B)=K_\ell\rtimes S_\ell$, where $K_\ell\triangleleft Z_{S_{2^\ell}}(B)$ is normal. 
\end{enumerate}
\label{Centralizers}
\end{prop}
\begin{proof}
The statement in $(a)$ is an immediate consequence of the initial assumption. The assertion in $(b)$ follows from the fact that the element $\zeta=\sum\limits_{1\leq i<j\leq n} (ij)$ is in the center of the group algebra $\mathbb{C}[S_d]$. The justification for $(c)$ arises from the observation that for $\varphi(h)$ with $h\in S_{d}$ to commute with $B=\sum\limits_{g\in K_\ell}\varphi(g)$, it must satisfy the condition $\sum\limits_{g\in K_\ell}\varphi(hgh^{-1})=\sum\limits_{g\in K_\ell}\varphi(g)$. In essence, this implies that $h$ possesses the capability to rearrange or flip the bits.
\end{proof}

\begin{rmk}
It is interesting to point out that $K_\ell\rtimes S_\ell$ is $W(B_\ell)$, the Weyl group for root system of type $B_\ell$.
\end{rmk}

\begin{cor}
The subgroup of symmetries of the mixer $H_M$ surpasses that of $B$. For instance, when $\ell = 2$, $W(B_2)$ equals the dihedral group of order $|D_4| = 8$, while $|S_4| = 24$. For $\ell = 3$, $|W(B_3)| = 48$, and $|S_8| = 8! = 40320$. As $\ell$ increases, the disparity in orders becomes more pronounced: $|W(B_\ell)| = 2^\ell\cdot \ell!$ and $|S_d| = 2^\ell!$.
\label{GroupCor}

\end{cor}

\end{document}